\documentclass[aip,amsmath,amssymb,reprint]{revtex4-1}

\usepackage{graphicx}
\usepackage{dcolumn}
\usepackage{bm}
\usepackage[utf8]{inputenc}
\usepackage[T1]{fontenc}
\usepackage{mathptmx}
\usepackage{etoolbox}

\usepackage{amsthm}
\usepackage{booktabs}
\usepackage{hyperref}
\usepackage[dvipsnames]{xcolor}
\usepackage{multirow}
\usepackage{bbm}
\usepackage[normalem]{ulem}
\usepackage{units}
\usepackage{placeins}
\usepackage{algorithm}
\usepackage{algpseudocode}
\usepackage{float}
\usepackage{tikz}
\usepackage{array}

\makeatletter
\def\@email#1#2{%
\endgroup
\patchcmd{\titleblock@produce}
{\frontmatter@RRAPformat}
{\frontmatter@RRAPformat{\produce@RRAP{*#1\href{mailto:#2}{#2}}}\frontmatter@RRAPformat}
{}{}
}%
\makeatother

\newtheorem{proposition}{Proposition}
\newtheorem{definition}{Definition}
\newtheorem{corollary}{Corollary}
\newtheorem{remark}{Remark}

\newtheorem{thm}{Theorem}
\newtheorem{lemma}[thm]{Lemma}

\newcommand{\tcom}[1]{\textcolor{black}{#1}}
\begin{document}

\title{Ordinal Pattern Transition Networks as a Symbolic Substrate for Reservoir Computing}

\author{Max Potratzki}
\affiliation{Department of Epileptology, University Hospital Bonn, Bonn, Germany}
\affiliation{Helmholtz Institute for Radiation and Nuclear Physics, University of Bonn, Bonn, Germany}

\author{Thomas G. de Jong}
\affiliation{Division of Business \& Technology, Temple University, Japan Campus, 1-14-29 Taishido, Setagaya-ku, Tokyo 154-0004, Japan}
\affiliation{School of Mathematical Sciences, Hebei Normal University, 20 Nan’erhuan East Road, Yuhua District, Shijiazhuang,Hebei 050024, China}

\author{Michael Small}
\affiliation{Department of Mathematics and Statistics, The University of Western Australia, Perth, Australia}
\affiliation{UWA Data Institute, The University of Western Australia, Perth, Australia}

\author{Timo Br\"ohl}
\email{timo.broehl@ukbonn.de}
\affiliation{Department of Epileptology, University Hospital Bonn, Bonn, Germany}

\date{\today}

\begin{abstract}
Reservoir computing has established itself as a computationally efficient alternative to fully trained recurrent neural networks, using a fixed high-dimensional nonlinear substrate to map temporal input into a feature space for a simple readout.
Here we argue that reservoir computing can be realized entirely within a coarse-grained, symbolic state space, and that constrained random walks on Ordinal Pattern Transition Networks provide a natural instance of this realization, replacing the reservoir's continuous state space with a discrete and directly data-derived one.
We develop a formal correspondence between the two frameworks, show that the echo state property has a natural discrete analogue in the ergodicity and spectral gap of the walk's underlying Markov chain, and illustrate the correspondence with a minimal example and an empirical one-step-ahead prediction task evaluated across five paradigmatic model systems. These systems exhibit dynamics that correspond to strictly periodic, quasi-periodic, chaotic or stochastic signals.
\tcom{We argue that this symbolic reservoir perspective allows a redefinition of the reservoir substrate itself, replacing the requirement of a continuous, randomly generated dynamical system with any sufficiently ergodic, data-derived symbolic process. This opens a route toward interpretable, data-adaptive reservoirs whose internal structure is inferred rather than randomly assigned, with implications for surrogate data methods and reservoir design more broadly.}
\end{abstract}

\maketitle

\begin{quotation}
Reservoir computing, a form of machine learning designed for time-dependent data, uses a fixed dynamical system to transform an input signal into a set of internal features, after which only an output rule must be learned.
We show that this idea need not rely on a conventional and continuous-valued system.
A time series can instead be converted into a sequence of ordinal patterns -- local orderings of amplitude values.
Ordinal patterns and the observed transitions between them together define a network.
A random walk constrained by this data-derived network provides a discrete, symbolic analogue of a reservoir.
We establish conditions under which the walk progressively loses dependence on its initial state -- the symbolic counterpart of the fading-memory property required for reliable reservoir computing -- across systems ranging from periodic and quasiperiodic behavior to chaotic and stochastic dynamics.
This turns an ansatz used to describe complex signals into a computational substrate whose structure is interpretable and inferred directly from data rather than chosen at random, harboring potential for dynamics-informed reservoir design and for proxy- and surrogate-data generation.
\end{quotation}

\section{Introduction}

Reservoir computing (RC) conceptualizes temporal information processing as a three-stage architecture (cf. Fig.~\ref{fig:optn_esn_schematic} right): an input layer that injects a signal into a fixed nonlinear dynamical system (the reservoir), the reservoir itself, whose high-dimensional internal state nonlinearly encodes the signal's recent history, and a simple, typically linear readout that is the only trained component of the system~\cite{yan2024emerging}.






In its most common discrete-time instantiation, reservoir computing is realized as
\begin{equation}
\begin{aligned}
x(t+1) &= (1-\gamma)x(t) + \gamma f\big(Wx(t) + W^{(\text{in})}u(t) + b\big), \\
y(t) &= W^{(\text{out})}x(t),
\end{aligned}
\label{eq:rc_discrete}
\end{equation}
where $u(t)$ is the input, $x(t)$ the internal reservoir state, $y(t)$ the output, $W$ the fixed recurrent matrix, $W^{(\text{in})}$ the input-to-reservoir coupling, and $W^{(\text{out})}$ the trained readout~\cite{jaeger2002adaptive,lukovsevivcius2009reservoir,yan2024emerging}.
In practice, the reservoir state is typically augmented with a constant bias component so that the readout can represent affine mappings~\cite{lukovsevivcius2009reservoir}. 
We implicitly adopt this convention here: writing $x(t)$ for the reservoir state already includes a constant entry (e.g., $x(t) = [x_{\mathrm{res}}(t); 1]$), and $W^{(\text{out})}$ acts on this bias-augmented state. For brevity, we keep the compact notation $y(t) = W^{(\text{out})}x(t)$ throughout.
Overall this is the standard echo state network (ESN) formulation and serves here as our canonical discrete-time example.

More generally, the RC paradigm can be expressed as a coupled continuous-time dynamical system
\begin{equation}
\dot{x} = F(x, u, p), \qquad y = G(x, u, q),
\label{eq:rc_general}
\end{equation}
where $u$ is the input, $x$ the internal reservoir state, $y$ the output, $p$ the fixed reservoir parameters, and $q$ the trainable readout parameters~\cite{yan2024emerging,tanaka2019recent}. Equations~\eqref{eq:rc_discrete} and~\eqref{eq:rc_general} are two instances of the same reservoir-computing architecture: a fixed nonlinear dynamical substrate with fading memory, together with a separately trained readout.

The matrix $W$ -- random, fixed, and untrained -- defines the reservoir's internal coupling structure, $W^{(\text{in})}$ maps $u(t)$ into the reservoir state, and only $W^{(\text{out})}$ is optimized, typically via ridge regression.\cite{yan2024emerging} This architectural minimalism is precisely what makes RC computationally attractive relative to deep learning: no backpropagation is required, and the parameter count is orders of magnitude smaller than comparable deep recurrent networks~\cite{tanaka2019recent,yan2024emerging}.

For the purposes of this paper, we take a reservoir computer to require three architectural ingredients, independent of substrate: (i) a fixed, untrained, nonlinear state-transition map, (ii) a notion of fading memory / echo state property, and (iii) a separately trained readout mapping states.
This definition deliberately excludes physical or algorithmic details (recurrent-matrix construction, hardware substrate, or whether the update is deterministic or stochastic), and is broad enough to include Next Generation Reservoir Computing~\cite{gauthier2021next,cestnik2026next}, swarm RC~\cite{lymburn2021reservoir,lund2026state}, and physical reservoirs as special cases~\cite{nakajima2015information,tanaka2019recent,wang2025optical,de2025harnessing,de2025designing,ludge2026time}, while still being restrictive enough to exclude, e.g., arbitrary fixed feature maps lacking any temporal/history dependence.

The generality of Eqs.~\eqref{eq:rc_discrete}--\eqref{eq:rc_general} raises a natural question: must the reservoir be a continuous-valued dynamical system realized by a random matrix or a physical substrate, or can other classes of fixed, nonlinear, high-dimensional transformations serve equally well as reservoirs?
Recent theoretical work has already extended RC substrates toward ``structured graphs with random as well as non-random weights''~\cite{yan2024emerging}, explored reservoirs of specific topologies~\cite{kawai2019small}, and emphasized nonlinear dynamics, physical
realizations, and reservoir topology as central design considerations in reservoir computing~\cite{amann2026nonlinear}.
Yet almost all existing approaches -- whether the topology is random, layered, or task-optimized -- retain a continuous state vector evolving under a differentiable nonlinearity.
Recent work on stochastic reservoir computers and on stochastic state-space systems has further broadened this picture, modeling reservoirs as controlled Markov processes with probabilistic state transitions and providing a unified treatment of echo-state and fading-memory properties in both deterministic and stochastic settings.\cite{ehlers2025stochastic,ortega2026stochastic}.

Here we propose and formally develop an alternative: a \emph{symbolic reservoir} built directly from the ordinal (rank-order) structure of an observed time series, in which the fixed substrate is not a randomly generated matrix but an \emph{empirically inferred} transition network, and the reservoir's evolution is not a continuous nonlinear map but a \emph{constrained stochastic walk} over a finite symbolic state space. 
We refer to this substrate as an Ordinal Pattern Transition Network (OPTN)~\cite{small2013complex}. OPTNs lie within the broader ordinal-methods framework, which uses rank-order representations to investigate the structure and dynamics of complex systems and time series\tcom{~\cite{amigo2023ordinal,voltarelli2024characterizing}}. 
Recent research has established that OPTNs, derived from suitable observables of complex
dynamical systems---both model and real-world---encode essential dynamical properties of the underlying system~\cite{kulp2016using,yang2020multivariate,olivares2020contrasting,huang2021characterizing, potratzki2026quantifying,leo2026ordinal}.

Our contribution in this perspective is conceptual and structural rather than methodological: we do not present a new algorithm for challenges usually tackled with reservoir computing, but rather argue that reservoir computing's defining computational paradigm does not require a continuous state space, and that OPTN-based \emph{"random"} walks provide a natural symbolic instance of it, satisfying, almost without modification, the formal requirements of a reservoir computer. What distinguishes this construction from an ordinary higher-order symbolic Markov predictor is not the update rule itself -- at the level of a single prediction step, the two coincide -- but the fact that the OPTN walk satisfies the Echo State Property (ESP) and fading-memory criteria (Section~\ref{sec:fading_memory}) that formally define RC as a computational paradigm, independent of whether the underlying state space is continuous or symbolic.
In this sense, our OPTN-based walk can be viewed as a specific, data-derived symbolic instance of the broader class of stochastic reservoir computers and Markov-chain-based reservoirs~\cite{grigoryeva2018universal,gonon2019reservoir,ortega2026stochastic}, with the additional structure that its transition topology is inferred from ordinal statistics of a single observed time series rather than parametrically specified or randomly initialized.
We believe this recognition is valuable independent of any specific application, because it (i) imports RC's mature theoretical vocabulary -- echo state property, fading memory, memory capacity -- into the ordinal-pattern and symbolic-dynamics literature, and (ii) suggests that data-derived, non-random reservoirs constitute an underexplored point in RC's design space.

\section{Minimal Formal Setup}

\subsection{Ordinal patterns}

Given a scalar time series $\{s_1, s_2, s_3, \dots\}$, the ordinal pattern of order $d$ at time $t$ is the permutation $\pi_t \in S_d$ (the symmetric group on $d$ elements) describing the relative rank order of the $d$ values $\{s_t, s_{t+\tau}, \dots, s_{t+(d-1)\tau}\}$~\citep{bandt2002permutation}.
Each of the $d!$ possible permutations defines a distinct ordinal pattern, which serves as a symbol in the resulting symbolic description. The mapping from a real-valued signal to the ordinal-pattern sequence $\{\pi_t\}$ discards amplitude information while retaining the local ordinal (shape) structure of the dynamics.
Throughout, we use \emph{ordinal pattern} when referring to a specific permutation $\pi_t\in S_d$, and \emph{symbolic} when referring more generally to the discrete representation or dynamics induced by the ordinal encoding.

\subsection{The ordinal pattern transition network}

The OPTN is the weighted directed graph on $d!$ nodes, each associated with an ordinal pattern, whose edge weight $A_{ij}$ counts (or estimates the probability of) observing ordinal pattern $j$ immediately following ordinal pattern
$i$ in the empirical ordinal-pattern sequence~\citep{mccullough2015time}.
Row-normalizing $A$ yields a transition probability matrix $T$, and the pair $(A,T)$ constitutes a complete, data-derived, fixed nonlinear structure: fixed because it is built once from a training signal and not subsequently updated; nonlinear because the mapping from real-valued input to symbolic transitions is inherently non-invertible and rank-based rather than linear.

\subsection{A constrained random walk on the OPTN}
\label{sec:walk}

A random walk on the OPTN generates a new ordinal-pattern sequence
$\{\hat\pi_t\}$ by sampling successor states according to $T$, subject
to an ordinal-compatibility constraint~\citep{sakellariou2019markov}.
The walk advances through pattern positions one step at a time, i.e.,
$\pi_t \to \pi_{t+1} \to \pi_{t+2} \to \cdots$; because the OPTN
transition matrix $T$ is estimated empirically from adjacent-by-one
pattern pairs, the edge $\pi_t \to \pi_{t+1}$ captures no raw-sample
overlap by itself when $\tau>1$: adjacent-by-one pattern pairs $\pi_t$
and $\pi_{t+1}$ generally share no samples at all in that case. The
additional logical constraint instead comes from patterns separated
from the candidate by a multiple of $\tau$: since two patterns
separated by $\tau$ share exactly $d-1$ overlapping sample values, the
legality of $\pi_{t+1}$ is constrained by the $\tau$-shifted patterns
\[
\big(\pi_{t+1-(d-1)\tau},\ \dots,\ \pi_{t+1-2\tau},\ \pi_{t+1-\tau}\big),
\]
a tuple of $d-1$ patterns already generated at earlier walk steps.
Because each subsequent step $t+2, t+3,\dots$ requires a \emph{different}
$\tau$-shifted sub-selection from this same neighborhood, no fixed
$\tau$-spaced tuple can serve as an incrementally updatable state
across multiple steps. What does suffice is a \emph{contiguous} window
of the $h=(d-1)\tau$ most recently generated patterns,
\[
\big(\pi_{t+1-(d-1)\tau},\ \dots,\ \pi_t\big),
\]
which we call the \emph{compatibility state} at step $t$ (governing
the transition to $t+1$) \tcom{(see Fig.~\ref{fig:optn_esn_schematic}, right panel, for a graphical illustration of this overlap-induced compatibility constraint for $d = 3$, $\tau = 5$).
It contains every $\tau$-shifted sub-tuple that any future step within this window will need to check compatibility against, and it updates by simply dropping its oldest entry and appending the newly generated pattern as the walk advances.}

For a compatibility state
\[
\sigma_t = \big(\pi_{t-h+1},\ldots,\pi_t\big),
\]
let $\mathcal{C}(\sigma_t)$ denote the set of ordinal patterns $j$
that are both reachable from the current terminal pattern $\pi_t$ in
the OPTN and compatible with the $\tau$-shifted sub-tuples encoded by
$\sigma_t$. The constrained walk samples its next pattern from the
conditional distribution
\begin{equation}
\Pr\!\big(\pi_{t+1}=j \mid \sigma_t\big)
=
\frac{T_{\pi_t j}}
{\sum_{\ell\in\mathcal{C}(\sigma_t)} T_{\pi_t \ell}},
\qquad
j\in\mathcal{C}(\sigma_t),
\label{eq:compatibility_conditioned_transition}
\end{equation}
and assigns probability zero to $j\notin\mathcal{C}(\sigma_t)$.
Thus, the original pattern-level transition matrix $T$ supplies the
relative successor weights, while the compatibility constraint
conditions these weights on the set of legally admissible
successors. In what follows, we restrict attention to states for
which $\sum_{\ell\in\mathcal{C}(\sigma_t)}T_{\pi_t\ell}>0$.

We emphasize that this compatibility state -- a contiguous block of
$h=(d-1)\tau$ ordinal patterns -- is a distinct object from the
\emph{seed sequence} needed to initialize a free-running walk from an
observed time series (Table~\ref{tab:correspondence}), which must
additionally supply enough raw samples to compute this initial block of ordinal patterns in the first place.

The walk is therefore not first-order Markov on the space of single
patterns; rather, its natural state is this compatibility-state
window, and it becomes first-order Markov once formulated on this
enlarged state space (made precise in
Definition~\ref{def:optn_chain}, Section~\ref{sec:fading_memory}).
This is the generic object of comparison for the remainder of this
paper; we deliberately do not specify the amplitude-reconstruction
step that converts the symbolic walk back into a scalar series.

\section{Central Correspondence: OPTN-Based Random Walk vs. RC}
\label{sec:central_correspondence}

We now state the mapping between Eq.~\eqref{eq:rc_discrete} and the OPTN random-walk explicitly (Table~\ref{tab:correspondence} and Figure~\ref{fig:optn_esn_schematic}).

\begin{table*}[htbp]
\centering
\caption{Correspondence between the standard discrete-time echo state network (ESN) formulation of reservoir computing in Eq.~\eqref{eq:rc_discrete} and constrained OPTN random walks.
The ESN / RC column lists the components of this canonical discrete-time reservoir; other RC substrates (continuous-time, physical, stochastic) share the same architectural roles but differ in implementation.}
\label{tab:correspondence}
{\renewcommand{\arraystretch}{1.15}
\begin{tabular}{@{}lll@{}}
\toprule
\textbf{RC/ESN component (Eq.~\eqref{eq:rc_discrete})} & \textbf{OPTN walk analogue} & \textbf{How OPTN differs} \\
\midrule
\raggedright\parbox[t]{4.5cm}{\raggedright Fixed reservoir matrix $W$} & \parbox[t]{6.5cm}{\raggedright OPTN transition matrix $T$ (or count matrix $A$)} & \parbox[t]{5.5cm}{\raggedright Data-derived rather than randomly generated} \\
\parbox[t]{4.5cm}{\raggedright Internal state $x(t)\in\mathbb{R}^{D_r}$} & \parbox[t]{6.5cm}{\raggedright Current compatibility state: the contiguous window of the $h=(d-1)\tau$ most recently generated ordinal patterns that constrains the legal successors (Definition~\ref{def:optn_chain})} & \parbox[t]{5.5cm}{\raggedright Discrete, symbolic state rather than continuous vector} \\
\parbox[t]{4.5cm}{\raggedright Nonlinear update $f(Wx+W^{(\text{in})}u+b)$} & \parbox[t]{6.5cm}{\raggedright Compatibility-conditioned stochastic transition: candidate successors restricted by OPTN edges and ordinal compatibility, with their transition weights renormalized over the legal successor set} & \parbox[t]{5.5cm}{\raggedright Stochastic sampling replaces deterministic nonlinear activation} \\
\parbox[t]{4.5cm}{\raggedright Driven (listening / washout) phase} & \parbox[t]{6.5cm}{\raggedright Seed sequence drawn from the real signal, of sufficient length to compute the initial compatibility-state window of $h=(d-1)\tau$ ordinal patterns} & \parbox[t]{5.5cm}{\raggedright Functions as washout / synchronization transient} \\
\parbox[t]{4.5cm}{\raggedright Autonomous / closed-loop phase (digital twin)} & \parbox[t]{6.5cm}{\raggedright Free-running walk over compatibility states after the seed is consumed} & \parbox[t]{5.5cm}{\raggedright No external input required, exactly as required of an RC-based dynamical twin \cite{kong2023reservoir}} \\
\parbox[t]{4.5cm}{\raggedright Trainable readout $W^{(\text{out})}$} & \parbox[t]{6.5cm}{\raggedright A fixed, parameter-free maximum-likelihood mapping (Section~\ref{sec:pred_rule}) from compatibility state to predicted successor ordinal pattern} & \parbox[t]{4.5cm}{\raggedright A first, minimal instance of a symbolic readout; trained parametric readouts remain an open extension (Section~\ref{sec:discussion})} \\
\parbox[t]{4.5cm}{\raggedright Reservoir dimension $D_r$} & \parbox[t]{6.5cm}{\raggedright Compatibility state-space size $\Sigma_h \subseteq \{1,\dots,d!\}^{h} \leq (d!)^{h}$, with $h=(d-1)\tau$} & \parbox[t]{5.5cm}{\raggedright Grows combinatorially with embedding dimension $d$ and delay $\tau$; strict inequality reflects compatibility-inadmissible states} \\
\bottomrule
\end{tabular}
}
\end{table*}

\begin{figure*}[htbp]
\centering \includegraphics[width=2\columnwidth]{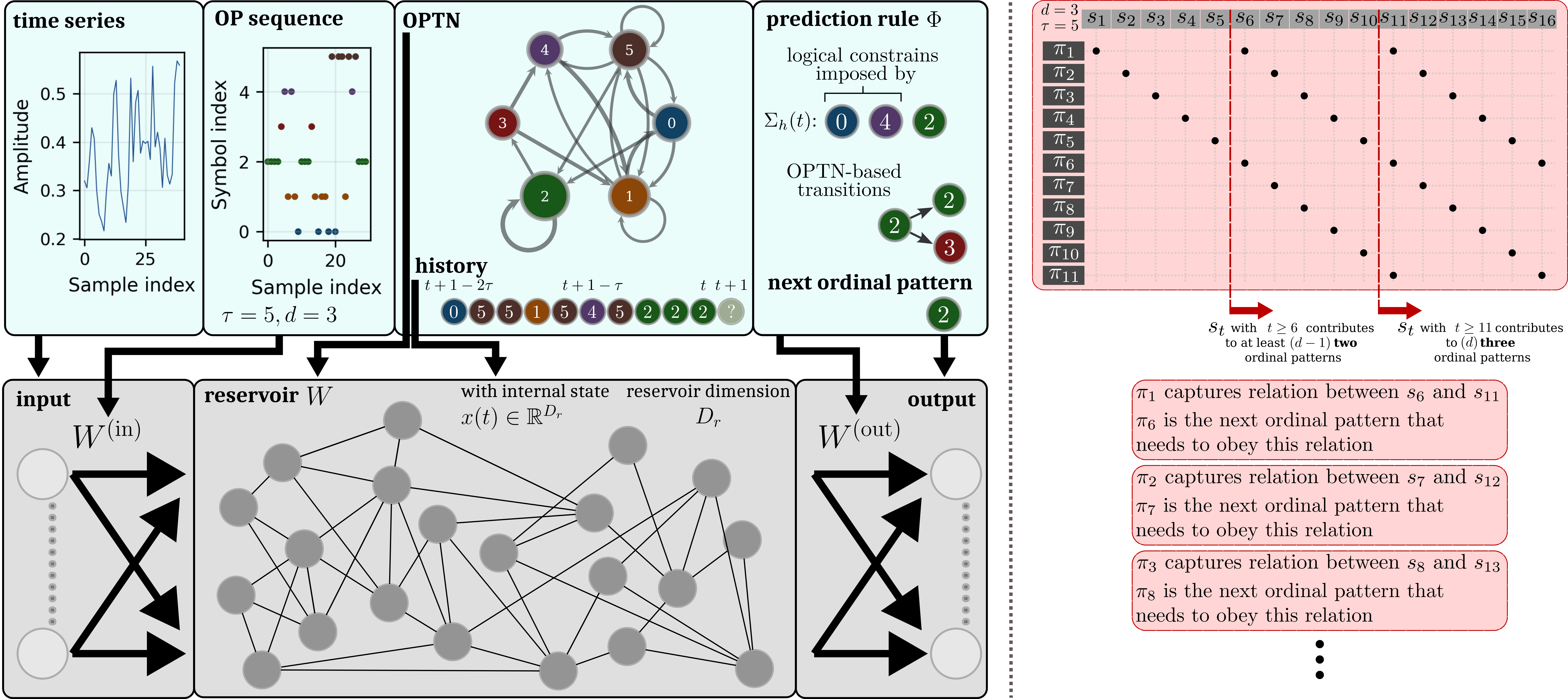}
\caption{{\bf Left panel:} Conceptual correspondence between the OPTN
symbolic walk (top left) and a classical echo state network (bottom
left). Top left: a time series is encoded into ordinal patterns, which
populate the OPTN; the compatibility state $\Sigma_h(t)$
(Definition~\ref{def:optn_chain}) and prediction rule $\Phi$
(Section~\ref{sec:pred_rule}) determine the next ordinal pattern. Bottom right:
the analogous reservoir computing pipeline, where the internal state
$x(t)$ and readout $W^{(\text{out})}$ produce the output. Arrows indicate the
correspondences of Table~\ref{tab:correspondence}: $T \leftrightarrow
W$ and $\Sigma_h(t) \leftrightarrow x(t)$. These are structural
analogies, not exact equivalences (Section~\ref{sec:partial}).
{\bf Right panel:} Illustration of the ordinal-compatibility constraint
underlying Section~\ref{sec:walk} for $d=3$, $\tau=5$. Each
ordinal pattern $\pi_i$ is computed from $d$ samples spaced $\tau$
apart (e.g., $\pi_1$ from $s_1, s_6, s_{11}$); because patterns
separated by $\tau$ share $d-1$ overlapping samples, a given sample
$s_t$ generally contributes to several patterns simultaneously, and
the values it has already fixed in earlier patterns constrain which
ordinal patterns can legally follow. Once $t \geq (d-1)\tau + 1$
(here $t \geq 11$), each new sample contributes to $d$ patterns
still under construction, so the next pattern is not free to take any
of the $d!$ possible values, but must remain consistent with the
rank-order relations already fixed by the preceding $\tau$-shifted
patterns. This overlap-induced dependency is the source of the
compatibility constraint restricting admissible successor patterns in
the OPTN walk (Section~\ref{sec:walk}); the full contiguous window of
$h=(d-1)\tau$ patterns needed to track this dependency across
successive walk steps motivates the compatibility-augmented state
space $\Sigma_h$ of Definition~\ref{def:optn_chain}.}
\label{fig:optn_esn_schematic}
\end{figure*}

Two aspects of this correspondence merit emphasis. 
First, the OPTN walk still requires two user-specified hyperparameters, the embedding dimension $d$ and delay $\tau$, which jointly determine the resolution and admissible-state-space size $\Sigma_h$ of the substrate. Unlike continuous RC, however, these hyperparameters do not need to be tuned against a stability criterion: there is no analogue of a spectral-radius or leak-rate search~\cite{jaeger2002adaptive,jaeger2007optimization}, since the walk's transition structure $T_h$ is read off directly from the data once $(d,\tau)$ are fixed, rather than requiring a separate random-matrix-generation and validation step to ensure the ESP holds. 
What \emph{is} still required for the discrete echo-state property to hold are the usual graph-level conditions of the induced Markov chain on compatibility states (Definition~\ref{def:optn_chain}): irreducibility and aperiodicity of $T_h$ on the closed communicating class $\Sigma_h$. In practice, these conditions correspond to strong
connectivity and the absence of perfect periodic cycles on the
compatibility-state graph.
When these conditions fail -- e.g., for strictly periodic or constant signals inducing a periodic symbolic chain -- the walk does not have fading memory (cf. Proposition~\ref{prop:esp}), but its predictability is then governed by the trivial periodic structure rather than by spectral-gap mixing.
Second, the driven/autonomous distinction that Kong et al. identify as the defining requirement of RC-based digital twins -- ``the digital twin must be able to evolve its dynamical state at the present time to the next time step without further state input''~\cite{kong2023reservoir} -- is satisfied by construction once the walk's seed is consumed, without any additional architectural machinery.

\section{Fading Memory as Ergodicity}
\label{sec:fading_memory}

The echo state property is defined as the condition under which a reservoir's asymptotic state depends only on the input sequence and not on the initial network state, a condition closely tied to the fading memory property in which the current state depends predominantly on recent, rather than distant, history \cite{yan2024emerging}.
For continuous RC, Jaeger's original sufficient condition requires the largest singular value of $W$ to be less than one, with the spectral radius providing a necessary condition ruling out asymptotically unstable reservoirs \cite{yan2024emerging}.

We now formalize the discrete analogue of this property for OPTN-based walks.
\tcom{Intuitively, the compatibility state introduced below is the minimal window of recently generated ordinal patterns required to determine which successor patterns remain legally admissible. }

\begin{definition}[OPTN Markov chain on compatibility-augmented states]
\label{def:optn_chain}

Let $h=(d-1)\tau$ denote the length of the contiguous window of
ordinal patterns relevant to the compatibility constraint
(Section~\ref{sec:walk}). A compatibility state is a contiguous tuple
of $h$ ordinal patterns,
\[
\sigma
=
\big(\pi_{t-h+1},\pi_{t-h+2},\ldots,\pi_t\big),
\]
whose $\tau$-shifted sub-tuples satisfy the overlapping-window
constraint of Section~\ref{sec:walk}.

For a compatibility state
\[
\sigma
=
\big(\pi_{t-h+1},\ldots,\pi_t\big),
\]
let
\[
i(\sigma)=\pi_t
\]
denote its terminal ordinal pattern, and write
$\mathcal{C}(\sigma)$ for the set of graph-reachable and ordinally
compatible successor patterns defined in Section~\ref{sec:walk}.
By the restriction stated there, only states for which the
normalizing denominator below is positive are considered.

The OPTN-constrained walk induces a discrete-time stochastic process
$\{X_n\}_{n\geq0}$, where $X_n$ is the compatibility state after the
$n$-th walk step. Its transition probabilities are defined by the
compatibility-conditioned transition matrix $T_h$:
\begin{equation}
\begin{aligned}
(T_h)_{\sigma\sigma'}
&=
\Pr(X_{n+1}=\sigma'\mid X_n=\sigma)
\\
&=
\begin{cases}
\displaystyle
\frac{T_{i(\sigma)j}}
{\sum_{\ell\in\mathcal{C}(\sigma)}T_{i(\sigma)\ell}},
&
\substack{
j\in\mathcal{C}(\sigma),\\
\sigma'=\big(\pi_{t-h+2},\ldots,\pi_t,j\big),
}
\\[1.4em]
0,
&
\text{otherwise}.
\end{cases}
\end{aligned}
\label{eq:induced_transition_matrix}
\end{equation}

Thus, at each step the walk drops the oldest pattern in the current
compatibility state and appends an admissible successor pattern $j$.
The original pattern-level transition matrix $T$ supplies the
relative successor weights, while the compatibility constraint
renormalizes these weights over the legal successor set
$\mathcal{C}(\sigma)$. Consequently, every row of $T_h$ sums to one
on the retained state space.

For the ergodic analysis below, we further restrict attention to a
closed communicating class of these states, still denoted by
$\Sigma_h$. On this class, $T_h$ is a finite row-stochastic transition
matrix.

Crucially, the transition probability from $\sigma$ to $\sigma'$ is
determined entirely by the current compatibility state $\sigma$, the
candidate successor state $\sigma'$, and the fixed pattern-level
matrix $T$. It makes no reference to states visited prior to
$\sigma$. The compatibility state therefore records, in a contiguous
and incrementally updatable form, all information needed to determine
the admissible successor states and their transition probabilities.

The induced stochastic process on $\Sigma_h$ is shown in
Lemma~\ref{lem:markov_property} to be a first-order homogeneous
Markov chain. We write $k=|\Sigma_h|$ for the number of states in the
closed communicating class considered below.

\end{definition}

\begin{remark}[Compatibility conditioning]
\label{rem:compatibility_conditioning}
The pattern-level matrix $T$ is row-stochastic before ordinal
compatibility is imposed. Compatibility may exclude some
graph-reachable successors of a given terminal pattern; therefore,
the induced history-state transition matrix $T_h$ is obtained by
conditioning the outgoing weights of $T$ on the legal successor set
$\mathcal{C}(\sigma)$. This conditioning is necessary for $T_h$ to
remain row-stochastic and hence to define a Markov chain on
compatibility states.
\end{remark}

Our ergodic treatment of the OPTN-induced Markov chain builds on the Markovian reformulation of ordinal partition networks introduced by
Sakellariou, Stemler \& Small~\cite{sakellariou2019markov}, who showed that recording transitions between ordinal patterns approximates the Perron-Frobenius operator in a higher-dimensional space topologically equivalent to the original state space, and who first formalized the observation -- central also to our Section~\ref{sec:walk} -- that this
symbolic chain is first-order Markov only once the pattern overlap satisfies $w \geq m$, and is otherwise a higher-order chain of order $m-w$ requiring a compatibility-augmented state description.
Their analysis uses this Markovian framework to connect network connectivity patterns to indicators of the underlying dynamics, most notably the hierarchy of unstable periodic orbits embedded within a chaotic attractor, and to derive practical guidance for choosing the pattern length and embedding lag.
The present work adopts the same compatibility-augmented Markov construction but directs it toward a different end: rather than using the chain's stationary and spectral properties to characterize or diagnose the
generating dynamics, we use the chain's ergodicity and spectral gap
(Proposition~\ref{prop:esp}, Corollary~\ref{cor:mixing_time}) to establish a formal, architectural correspondence with reservoir computing -- specifically, as the discrete analogue of the spectral-radius condition underlying the Echo State Property.
In this sense, the two contributions are complementary rather than competing: theirs establishes the ergodic-theoretic foundations of the ordinal Markov chain and its relationship to the source dynamics, while ours repurposes that same foundation to argue that the resulting symbolic substrate satisfies the formal requirements of a reservoir computer.

It is worth being explicit about why irreducibility of the compatibility-augmented chain, on its own, is not sufficient to establish this correspondence. 
First, irreducibility alone does not exclude periodicity: a periodic irreducible chain still possesses a unique stationary distribution $\mu$, but the distribution $\nu_0 T_h^n$ at step $n$, started from an arbitrary initial distribution $\nu_0$ over $\Sigma_h$, need not converge to $\mu$ at all -- it can cycle indefinitely among a fixed set of distributions, a behavior manifestly incompatible with a fading-memory property, in which the influence of the seed history must vanish over time rather than persist as sustained oscillation.
Aperiodicity is therefore required as an independent hypothesis, not a redundant strengthening of irreducibility. Second, even irreducibility together with aperiodicity only guarantees qualitative convergence, $\nu_0 T_h^n \to \mu$, via the standard ergodic theorem for Markov chains; it provides no statement about the \emph{rate} of this convergence. Because our fading-memory claim -- and, more concretely, the symbolic mixing-time diagnostic of Corollary~\ref{cor:mixing_time} -- requires a quantitative decay rate rather than mere eventual convergence, we additionally need the spectral gap of $T_h$ and, in the non-diagonalizable case, the Jordan-block structure governing sub-dominant eigenvalue decay. This is the source of the additional technical machinery in the proof of Proposition~\ref{prop:esp} below, relative to a bare irreducibility argument.

\begin{lemma}[Markov property of the compatibility state process]
\label{lem:markov_property}
Let $\Sigma_h$ denote the closed communicating class of admissible
compatibility states specified in Definition~\ref{def:optn_chain}, and
let
\[
X_n \in \Sigma_h
\]
be the compatibility state after walk step $n$ (Definition~\ref{def:optn_chain}).
Then the stochastic process $\{X_n\}_{n\ge0}$ is a first-order homogeneous Markov chain on $\Sigma_h$, i.e.,
\[
\Pr(X_{n+1}=x' \mid X_n=x, X_{n-1},\ldots,X_0)
=
\Pr(X_{n+1}=x' \mid X_n=x)
\]
for all admissible states $x,x'\in\Sigma_h$.
\end{lemma}

\begin{proof}
By Definition~\ref{def:optn_chain}, the transition rule $(T_h)_{\sigma\sigma'}$ is defined purely as a function of the ordered pair $(\sigma,\sigma')$: $\sigma'$ must be a valid one-step shift-and-append extension of $\sigma$, and its probability is obtained by conditionally normalizing the outgoing weight from the terminal pattern of $\sigma$ over the compatibility-admissible successor set $\mathcal{C}(\sigma)$. Both the admissibility decision and the normalization depend only on $\sigma$, $\sigma'$, and the
fixed pattern-level matrix $T$, with no reference to any state visited prior to $\sigma$. 
Because $\sigma$ is a contiguous window containing every $\tau$-shifted sub-tuple needed to evaluate the compatibility of any candidate successor pattern (Section~\ref{sec:walk}), every admissible successor state and its associated transition probability are therefore determined uniquely by the current compatibility state, independent of how $\sigma$ itself was reached. Therefore,
\[
\Pr(X_{n+1}\mid X_n,X_{n-1},\ldots)
=
\Pr(X_{n+1}\mid X_n),
\]
which is precisely the Markov property.
Since the transition probabilities do not depend explicitly on the walk step $n$, the chain is homogeneous.
\end{proof}

\begin{proposition}[Discrete fading memory via ergodicity]
\label{prop:esp}
By Lemma~\ref{lem:markov_property}, the compatibility state process
$\{X_n\}_{n\geq0}$ forms a finite homogeneous Markov chain on the
closed communicating class $\Sigma_h$ specified in
Definition~\ref{def:optn_chain}. Assume that this chain is
irreducible and aperiodic.
Then there exists a unique stationary distribution $\mu$ on $\Sigma_h$, and for any two initial \emph{distributions} $\nu_0, \nu_0'$ over $\Sigma_h$ (e.g., corresponding to two distinct one-shot seed compatibility states $X_0=\sigma$ versus $X_0=\sigma'$, so that $\nu_0=e_\sigma$ and $\nu_0'=e_{\sigma'}$ are point masses rather than the walk's realized trajectory itself),
\begin{equation}
\big\|\nu_0 T_h^n - \nu_0' T_h^n\big\|_{TV} \;\leq\; C\, n^{\,m^\star-1}\,\lambda_\star^{\,n}
\label{eq:mixing_bound}
\end{equation}
for a constant $C = \sqrt{k}\,\kappa(P)$, where $\lambda_\star = \max_{i\geq 2}|\lambda_i|<1$ is the second-largest eigenvalue modulus of $T_h$, $m^\star$ is the size of the largest Jordan block among the sub-dominant eigenvalues, and $\kappa(P)=\|P\|\|P^{-1}\|$ is the condition number of the eigenvector matrix in the Jordan decomposition $T_h=PJP^{-1}$.
Here the right eigenvector of $T_h$ associated with the leading eigenvalue $\lambda_1=1$ is the all-ones vector $\mathbf{1}$ (since $T_h$ is row-stochastic), and the corresponding left eigenvector is the stationary distribution $\mu$, so that $\mathbf{1}\mu$ denotes the rank-one projector onto this leading mode.
When $T_h$ is diagonalizable (the generic case for an empirically estimated matrix with distinct eigenvalues), $m^\star=1$ and the bound reduces to a geometric decay $C\lambda_\star^n$.
Consequently, the walk's state distribution at time $n$ becomes asymptotically independent of its initial (seed) compatibility state at a rate governed by the spectral gap $\gamma=1-\lambda_\star$.
\end{proposition}

\begin{proof}
\textbf{Step 1 (Spectral decomposition).}
Since $T_h$ is row-stochastic, irreducible, and aperiodic, the Perron--Frobenius theorem guarantees a simple leading eigenvalue $\lambda_1=1$ with left eigenvector $\mu$ (the stationary distribution, $\mu T_h = \mu$, $\mu>0$ entrywise), and all remaining eigenvalues satisfy $|\lambda_i|<1$.
Simplicity of $\lambda_1$ implies its Jordan block is exactly $1\times1$, so no polynomial correction arises from the dominant term.
Writing the Jordan decomposition $T_h = PJP^{-1}$, with $J$ block-diagonal consisting of the scalar block $[1]$ and Jordan blocks $J_i(\lambda_i)$ of size $m_i$ for $|\lambda_i|<1$, we have
\[
T_h^n = PJ^nP^{-1}, \qquad J^n = \mathrm{diag}\big(1,\ J_2(\lambda_2)^n,\dots\big).
\]
By the standard bound on powers of Jordan blocks~\cite[Ch.~3]{horn2012matrix},
\begin{equation*}
\begin{aligned}
\|J_i(\lambda_i)^n\| &\leq \sum^{m_i-1}_{j=0}\binom{n}{m_i-1}|\lambda_i|^{\,n-m_i+1} \|J_i(0)^j\|\\ &\leq C'_i\binom{n}{m_i-1}|\lambda_i|^{\,n-m_i+1}\\ &\leq C_i\, n^{m_i-1}|\lambda_i|^n
\end{aligned}
\end{equation*}
for a constant $C_i$ depending only on $m_i$ and the chosen matrix norm. Letting $m^\star=\max_i m_i$ and $\lambda_\star=\max_{i\geq2}|\lambda_i|$,
\begin{equation}
\big\|T_h^n - \mathbf{1}\mu\big\| \;\leq\; \kappa(P)\, n^{m^\star-1}\,\lambda_\star^{\,n},
\label{eq:jordan_bound}
\end{equation}
where $\mathbf{1}\mu$ is the rank-one projector onto the stationary distribution. Equation~\eqref{eq:jordan_bound} holds without assuming diagonalizability of $T_h$; when $T_h$ is diagonalizable, $m^\star=1$ and \eqref{eq:jordan_bound} reduces to a geometric bound.
Because $T_h$ is real, any complex eigenvalues occur in conjugate pairs and the Jordan decomposition, though carried out over $\mathbb{C}$, yields a real matrix $T_h^n$ for every $n$; the bound in \eqref{eq:jordan_bound}, stated in terms of eigenvalue moduli, is unaffected by this and holds as a bound on a real quantity throughout.

\textbf{Step 2 (Reduction to the difference of distributions).}
Let $\delta_0 = \nu_0-\nu_0'$. Since both $\nu_0$ and $\nu_0'$ are probability distributions, $\delta_0\mathbf{1}=0$, so $\delta_0$ has no component along the stationary projector: $\delta_0\,\mathbf{1}\mu = 0$. Hence
\[
\delta_0 T_h^n = \delta_0\big(T_h^n - \mathbf{1}\mu\big).
\]

\textbf{Step 3 (Norm conversion to total variation distance).} Using the standard inequality $\|x\|_1 \le \sqrt{k}\,\|x\|_2$ in dimension $k = |\Sigma_h|$, sub-multiplicativity of the operator norm, and $\|\delta_0\|_2 \le \|\nu_0\|_2 + \|\nu_0'\|_2 \le 2$ (each probability vector satisfies $\|\nu\|_2 \le \|\nu\|_1 = 1$),
\begin{equation*}
\begin{aligned}
\|\nu_0 T_h^n - \nu_0' T_h^n\|_{TV} = \tfrac{1}{2}\|\delta_0 T_h^n\|_1 &\le \tfrac{1}{2}\sqrt{k}\,\|\delta_0 T_h^n\|_2\\ &\le \tfrac{1}{2}\sqrt{k}\,\|\delta_0\|_2\,\|T_h^n - \mathbf{1}\mu\|_2.
\end{aligned}
\end{equation*}
Combining with Eq.~\eqref{eq:jordan_bound} and $\|\delta_0\|_2 \le 2$ yields
\[
\|\nu_0 T_h^n - \nu_0' T_h^n\|_{TV} \;\le\; \sqrt{k}\,\kappa(P)\, n^{m^\star-1}\lambda_\star^n \;=:\; C\, n^{m^\star-1}\lambda_\star^n,
\]
which is Eq.~\eqref{eq:mixing_bound}, with $C := \sqrt{k}\,\kappa(P)$ defined as an explicit, computable upper-bound constant rather than the tightest possible one: equality is approached only in the worst case $\|\delta_0\|_2 = 2$, corresponding to seed distributions $\nu_0, \nu_0'$ with disjoint support, and $C$ strictly exceeds the tight constant otherwise.

\end{proof}

\begin{corollary}[Symbolic mixing time as a memory-depth diagnostic]
\label{cor:mixing_time}
Define the mixing time $t_{\mathrm{mix}}(\varepsilon) = \min\{n : \|\nu_0 T_h^n - \mu\|_{TV} \leq \varepsilon \text{ for all } \nu_0\}$. From Eq.~\eqref{eq:mixing_bound}, solving $Cn^{m^\star-1}\lambda_\star^n\leq\varepsilon$ asymptotically gives
\begin{equation}
t_{\mathrm{mix}}(\varepsilon) \;=\; \frac{\log(C/\varepsilon)}{\log(1/\lambda_\star)} \;+\; O\!\left(\frac{\log\log(1/\varepsilon)}{\log(1/\lambda_\star)}\right),
\end{equation}
where the correction term arises from the polynomial factor $n^{m^\star-1}$ and vanishes exactly when $T_h$ is diagonalizable ($m^\star=1$).
Here $n$ counts walk steps, and, per Definition~\ref{def:optn_chain}, each walk step advances the raw-sample index by $1$ (in pattern-position units) rather than by $\tau$; the walk's effective memory horizon -- the number of steps after which the influence of the seed compatibility state becomes negligible -- is therefore $t_{\mathrm{mix}}(\varepsilon)$ pattern positions, directly computable from the spectral gap of the empirically estimated matrix $T_h$, with the compatibility-relevant window length $h=(d-1)\tau$ already absorbed into the definition of the state space $\Sigma_h$ itself (Definition~\ref{def:optn_chain}), rather than acting as an independent bound.
\end{corollary}

\begin{remark}[Finite-sample estimation error]
\label{rem:finite_sample}
Proposition~\ref{prop:esp} concerns the population-level chain with true transition matrix $T_h$. In practice, one estimates $\hat T_h$ from a finite symbolic sequence of length $N$, with each row a multinomial maximum-likelihood estimate over admissible successor compatibility states. For a state visited $N_\sigma$ times, the corresponding row of $\hat T_h$ satisfies, by the multivariate central limit theorem for multinomial proportions,
\[
\mathbb{E}\left[\|\hat T_h(\sigma,\cdot) - T_h(\sigma,\cdot)\|_2^2\right] = O(N_\sigma^{-1}),
\]
so that each row error is $O_p(N_\sigma^{-1/2})$. Assuming the empirical visitation counts are approximately balanced across the $k$ admissible states (i.e., $N_\sigma = \Theta(N/k)$ for all $\sigma$, a condition that holds approximately for the aperiodic, irreducible chains considered here once $N \gg k$), the Frobenius norm of the full perturbation matrix $E = \hat T_h - T_h$ satisfies
\[
\|E\|_F^2 = \sum_{\sigma} \|\hat T_h(\sigma,\cdot) - T_h(\sigma,\cdot)\|_2^2 = O_p\!\left(k \cdot \frac{k}{N}\right) = O_p\!\left(\frac{k^2}{N}\right),
\]
so that $\|E\|_2 \le \|E\|_F = O_p(k/\sqrt{N})$. Substituting into the Bauer--Fike bound,
\[
|\hat\lambda_\star - \lambda_\star| \;\le\; \kappa(P)\,\|E\|_2 \;=\; O_p\!\left(\kappa(P)\,\frac{k}{\sqrt{N}}\right),
\]
so that the estimated mixing time $\hat t_{\mathrm{mix}}(\varepsilon)$ computed from $\hat T_h$ is a consistent estimator of the true mixing time as $N \to \infty$, provided $N \gg k^2$. Because $k = |\Sigma_h|$ grows combinatorially with $d$ and $\tau$ (since $h=(d-1)\tau$), this estimation error grows sharply with embedding dimension and delay, reinforcing the practical trade-off already noted in Table~\ref{tab:correspondence}: larger $d$ or $\tau$ increases representational granularity at the cost of requiring substantially longer training sequences for reliable OPTN construction.
\end{remark}

\begin{remark}
Proposition~\ref{prop:esp} and Corollary~\ref{cor:mixing_time} together establish the discrete counterpart of the continuous ESP condition on the spectral radius of $W$~\cite{yan2024emerging}: in both cases, a spectral quantity of the fixed reservoir matrix (largest singular value/spectral radius for $W$; second-largest eigenvalue modulus for $T_h$) governs whether, and how fast, the reservoir's state forgets its initial or seed condition.
$\hat\lambda_\star$ for a finite, empirically estimated $\hat T_h$ is directly and exactly computable via standard eigenvalue decomposition -- requiring no heuristic search -- though it remains, per Remark~\ref{rem:finite_sample}, an estimate of the population quantity $\lambda_\star$ subject to finite-sample error.
\end{remark}

Our ergodicity–spectral-gap treatment thus specializes general fading-memory and universality results for stochastic reservoir computers and Markov-chain reservoirs~\cite{grigoryeva2018universal,gonon2019reservoir,ortega2026stochastic} to the particular case where the transition matrix arises from an OPTN.

\section{Illustrative Example}

To ground the correspondence concretely, consider a minimal case with embedding dimension $d=3$ (giving $3!=6$ possible ordinal patterns) and history length $h=d-1=2$, applied to a short noisy oscillatory signal.
Table~\ref{tab:toy} lists the six individual ordinal patterns together with illustrative pattern-to-pattern transition weights; these weights are the entries of $T$ and determine, via Definition~\ref{def:optn_chain}, the transition probabilities of the history-augmented chain on pairs $\sigma=(\pi_n,\pi_{n+1})\in\Sigma_2$.


\begin{table}[h]
\centering
\caption{Toy OPTN for $d=3$ and $\tau>1$. Transition weights (illustrative)
shown here are entries of the pattern-level matrix $T$; the formal walk of
Definition~\ref{def:optn_chain} operates on the induced chain over history
pairs $\Sigma_2$, not on single patterns directly. Right: graphical
illustration of the corresponding transition network.}
\label{tab:toy}

\begin{minipage}[c]{0.68\columnwidth}
\centering
\begin{tabular}{@{}cll@{}}
\toprule
Pattern ID & Rank order & Outgoing transitions \\
\midrule
0 & $(0,1,2)$  & $\to 1$ (0.6), $\to 4$ (0.4) \\
1 & $(0,2,1)$        & $\to 3$ (0.7), $\to 0$ (0.3) \\
2 & $(1,0,2)$        & $\to 5$ (0.5), $\to 2$ (0.5) \\
3 & $(1,2,0)$        & $\to 4$ (0.8), $\to 1$ (0.2) \\
4 & $(2,0,1)$        & $\to 2$ (0.6), $\to 3$ (0.4) \\
5 & $(2,1,0)$ & $\to 0$ (0.5), $\to 5$ (0.5) \\
\bottomrule
\end{tabular}
\end{minipage}%
\hspace{0.04\columnwidth}%
\begin{minipage}[c]{0.25\columnwidth}
\centering
\includegraphics[width=1.1\linewidth]{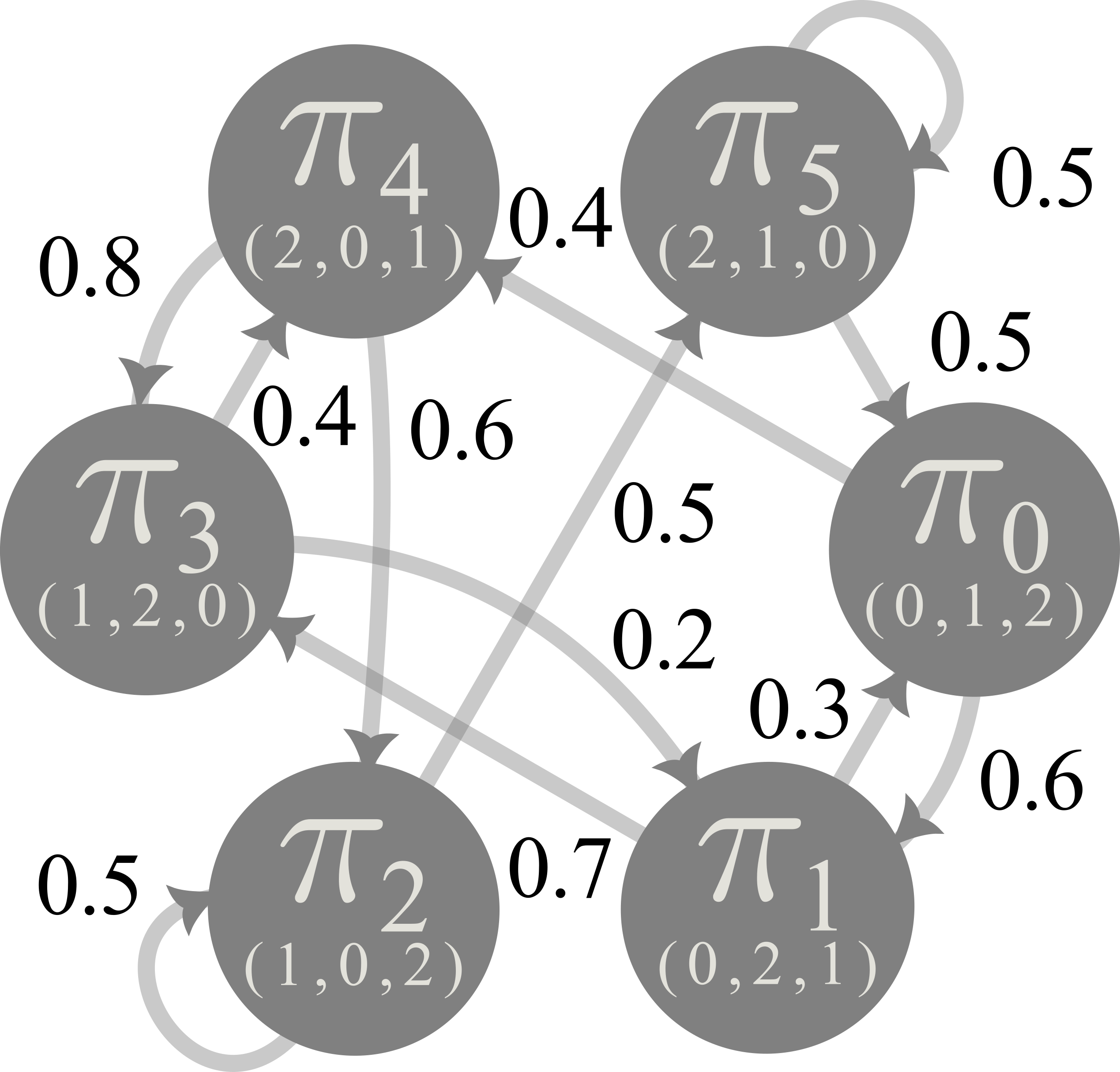}
\end{minipage}

\end{table}

A walk seeded at the history pair $\sigma_0=(0,1)$ (the ``driven'' phase, analogous to washout) evolves by appending one new pattern at a time and dropping the oldest: e.g., $(0,1)\to(1,3)\to(3,4)\to(4,2)\to(2,5)\to(5,0)\to\dots$, where each appended pattern is sampled from the outgoing weights of the \emph{most recent} pattern in the pair (Table~\ref{tab:toy}), conditionally renormalized over the successor patterns that satisfy the ordinal-compatibility requirement.
Reading this sequence of history pairs is functionally equivalent to reading off the successive internal states $x(t)$ of a conventional reservoir; the only difference is that these states are drawn from a finite, labeled alphabet $\Sigma_2$ rather than lying in a continuous vector space.

\section{OPTNs are not exactly RC}
\label{sec:partial}

Table~\ref{tab:correspondence} makes the correspondence explicit, but two structural asymmetries limit how far the analogy can be pushed.

\textbf{State-space and substrate origin.} Classical RC states occupy a continuous, typically high-dimensional real vector space, allowing arbitrarily fine gradations of representation, and the reservoir matrix $W$ is generated randomly, a design choice motivated by universal approximation results showing that sufficiently large random reservoirs can approximate broad classes of fading-memory filters without task-specific tuning~\cite{yan2024emerging,grigoryeva2018echo}.
This random-generation convention is not without alternatives, however: recent work has shown that task-specific, non-random wiring grown directly from a network's own activity statistics can outperform randomly initialized reservoirs~\cite{cazalets2026reshaping}.
The OPTN walk's state space $\Sigma_h$ is discrete and combinatorially bounded by $(d!)^h$, which caps representational granularity and, for large $d$, runs into severe data-sparsity problems in estimating $T_h$ reliably; and its substrate is \emph{directly inferred} from a specific training signal rather than randomly generated, making it intrinsically specialized to the data it was built from. This is a strength for signal-specific generative applications but a limitation for tasks requiring generalization across signal classes, an area where random RC reservoirs (and their associated universal approximation guarantees~\cite{grigoryeva2018echo,gonon2023approximation,ortega2026stochastic} currently have the stronger theoretical footing~\cite{yan2024emerging}.

\textbf{Operating mode and readout.} Most RC applications -- forecasting, classification, control -- are fundamentally driven: an external signal is continuously supplied, and the readout is trained against a task-specific target~\cite{lukovsevivcius2009reservoir,jaeger2004harnessing,yan2024emerging}. The OPTN walk, once past its brief seeding phase, operates entirely autonomously with no external driving signal, placing it closer -- in a conceptual sense -- to the ``digital twin'' mode of RC in which the reservoir evolves as a self-contained proxy for the original dynamical system~\cite{kong2023reservoir,lai2024digital,sun2026reservoir} than to conventional supervised RC pipelines. Section~\ref{sec:empirical} demonstrates that a simple, fixed (untrained) mapping from symbolic state to predicted
successor ordinal pattern is sufficient for one-step-ahead prediction, providing a first, minimal instance of a symbolic readout; whether a genuinely \emph{trained} readout -- e.g., a ridge or multinomial logistic regression over the history state -- can be formulated with efficiency guarantees comparable to Tikhonov-regularized linear RC readouts remains open, and is a natural direction for future work (Section~\ref{sec:discussion}).

\section{Empirical Illustration: One-Step-Ahead Prediction}
\label{sec:empirical}

Section~\ref{sec:fading_memory} established the OPTN walk's fading-memory property in terms of the spectral gap of the induced Markov chain $T_h$.
We now complement this analytical result with a minimal empirical probe of the walk's local predictive structure: given the true state of the walk at time $t$, how reliably does a symbolic prediction rule, built directly on the OPTN, recover the true successor ordinal pattern at $t+1$?
This section develops the evaluation methodology first, then reports results across five systems of increasing dynamical complexity -- from strictly periodic, over quasi-periodic, over chaotic (time-continous and discrete), to stochastic.
We emphasize that our purpose here is illustrative rather than exhaustive: we aim to demonstrate that the theoretical correspondence developed above translates into a concrete, testable procedure that mirrors the classical RC evaluation paradigm -- read out a prediction from held-out data using only the fixed substrate -- not to present a complete or optimized forecasting algorithm, and we deliberately restrict the empirical scope to one-step-ahead prediction to keep this demonstration self-contained and interpretable.

\subsection{Prediction rule}
\label{sec:pred_rule}

We define a deterministic symbolic predictor $\Phi$ that, given the walk's current state (the most recently visited pattern $\pi_t$ together with its compatibility-relevant history, Section~\ref{sec:walk}), returns a single most-likely successor pattern,
\begin{equation}
\hat\pi_{t+1} \;=\; \Phi(\pi_t; \, \text{history}) \;=\; \arg\max_{\pi' \,\in\, \mathcal{C}(t)} \hat T_h(\sigma_t, \sigma'),
\label{eq:pred_rule}
\end{equation}
where $\mathcal{C}(t)$ is the admissible set of candidate successor ordinal patterns defined in Section~\ref{sec:pred_eval}, $\sigma_t$ is the walk's current history state, $\sigma'$ is the history state obtained by appending candidate $\pi'$, and ties are broken uniformly at random.
Here $\hat T_h$ is the empirical counterpart of the compatibility-conditioned transition matrix in Eq.~\eqref{eq:induced_transition_matrix}, estimated from training counts and row-normalized over admissible successor compatibility states.
Concretely, $\Phi$ selects the admissible successor with the highest empirical transition probability under the OPTN's own estimated transition matrix $\hat T_h$: it is a fixed, parameter-free rule requiring no training beyond the OPTN's own empirically estimated transition counts, and no parameters beyond those already fixed by the substrate itself.
This is deliberately the simplest member of a broader family of symbolic predictors that could be built on the OPTN's transition structure -- for instance, a rule sampling proportionally to $\hat T_h$ rather than taking the arg max, or a trained parametric readout (Section~\ref{sec:discussion}) -- and we focus on this minimal instance here to isolate the substrate's own predictive structure from any effect of a more elaborate readout design.
Viewed in isolation, the predictor is a maximum-likelihood estimator on a higher-order Markov chain, which however can now be interpreted with regard to the history-augmented symbolic dynamics as a fixed reservoir whose structural and memory properties can be analysed with RC concepts.

\subsection{One-step-ahead prediction and evaluation metrics}
\label{sec:pred_eval}

To assess how well a symbolic prediction rule built on the OPTN captures the local transition structure of a signal, we evaluate one-step-ahead prediction: given the walk's true state at time $t$ (the current pattern $\pi_t$ together with its compatibility-relevant history), we compute the predicted successor $\hat\pi_{t+1}=\Phi(\pi_t;\text{history})$ and compare it against the true observed ordinal pattern $\pi_{t+1}$ in a held-out test segment.
This is evaluated at many randomly sampled positions across the test segment to obtain distributional estimates of prediction quality.
We quantify agreement between $\hat\pi_{t+1}$ and $\pi_{t+1}$ using complementary statistics, together with a chance-level reference against which they must be interpreted:

\begin{itemize}
\item \textbf{Exact accuracy}: the fraction of sampled test positions for which $\hat\pi_{t+1}=\pi_{t+1}$ exactly.
\item \textbf{Mean/median permutation distance}: since ordinal patterns are permutations in $S_d$, we quantify near-misses via $\rho(\hat\pi,\pi) = 1-\tau_K(\hat\pi,\pi)$, where $\tau_K \in [-1,1]$ is Kendall's rank correlation~\cite{kendall1938} between the two permutations. This distance vanishes for identical patterns, equals 1 in expectation for two independent (uncorrelated) patterns, and can exceed 1 (up to a maximum of 2) when predicted and true patterns are systematically \emph{reversed} relative to one another.
\item \textbf{Near-miss fraction}: the fraction of sampled positions with $\rho(\hat\pi,\pi)\leq0.3$, capturing predictions that are ordinally close without being exact.
\item \textbf{Corrected chance baseline}: rather than comparing exact accuracy against the naive uniform baseline $1/d!$ -- which implicitly assumes every one of the $d!$ ordinal patterns is a legal successor to any given state -- we compute, at each test position $t$, the size of the true \emph{admissible} candidate set $\mathcal{C}(t)$: the subset of patterns that are both graph-reachable from $\pi_t$ in the OPTN and ordinally compatible with the current history window (Section~\ref{sec:walk}). The corresponding chance-level accuracy at position $t$ is then $1/|\mathcal{C}(t)|$, and we report its average over the same sampled test positions used for the other three statistics. Because $|\mathcal{C}(t)|$ is generally far smaller than $d!$ and varies with $(d,\tau)$ and with the system's own dynamics, this corrected baseline -- rather than the naive uniform figure -- is the appropriate reference for judging whether a given exact accuracy or near-miss fraction reflects genuine predictive structure or merely the sparsity of the OPTN substrate itself.
\end{itemize}

\subsection{Model systems}
\label{sec:systems}

We evaluate the procedure of Sections~\ref{sec:pred_rule}--\ref{sec:pred_eval} on five systems \tcom{(Figure~\ref{fig:example_series})} chosen to span distinct mechanisms by which one-step-ahead predictability can, in principle, degrade, rather than a single linear scale of "regularity":

\begin{enumerate}
\item \textbf{Sine wave}: a purely periodic, single-frequency signal,
\begin{equation*}
x(t) = \sin(2\pi f t + \phi_0),
\end{equation*}
included as the strict-regularity limiting case. Because the signal is generated by a single fixed frequency $f= 5 \mathrm{Hz}$, the relationship between delay $\tau$ and the signal's period is exactly known, making this system a natural reference point for isolating sampling-related effects from genuinely dynamical ones.

\item \textbf{Coupled FitzHugh--Nagumo (FHN) oscillators}: a network of $N$ diffusively coupled, heterogeneous FitzHugh--Nagumo oscillators~\cite{fitzhugh1961impulses},
\begin{align*}
\dot{x}_i &= x_i(a-x_i)(x_i-1) - y_i + \frac{\kappa}{N-1}\sum_{j=1}^{N} A_{ij}(x_j - x_i), \\
\dot{y}_i &= b_i x_i - c y_i,
\end{align*}
with $A_{ij}$ being a random coupling structure, $N=100$, $\kappa=0.23$, $a=-0.02651$, $c = 0.02$ and $bi = 0.006 + \frac{i-1}{N-1}0.008$ ($0.006 \leq b_i\leq0.014 \forall i$).

It is a paradigmatic model for neural dynamics known to support rich collective behavior including extreme events and self-induced switching between multiple space-time patterns~\cite{ansmann2013extreme,mahata2025learning}.
The observable used here is the ensemble average $\langle x \rangle = \tfrac{1}{N}\sum_i x_i$.
We note explicitly that the qualitative regime realized by a given parameterization -- smooth and quasi-periodic, versus intermittently punctuated by rare, large-amplitude extreme events -- is a property of the specific operating point examined rather than of the FHN network in general, which is capable of substantially more irregular dynamics elsewhere in its parameter space.

\item \textbf{Lorenz system}: a canonical low-dimensional chaotic flow~\cite{lorenz1963deterministic},
\begin{align*}
\dot{x} &= \eta(y-x), \\
\dot{y} &= x(\xi - z) - y, \\
\dot{z} &= xy - \zeta z,
\end{align*}
with $\eta = 10$,
$\xi = 28$, $\zeta = 8/3$, used to probe predictability under genuine deterministic chaos -- sensitive dependence on initial conditions -- as distinct from periodicity, sampling artifacts, or injected noise. As a continuous-time flow, it additionally allows us to examine whether the numerical proximity of consecutive samples at small $\tau$, a property not shared by the discrete-time systems below, leaves a distinguishable signature in the OPTN-based prediction statistics. We chose $y$ as the observable for further investigations.

\item \textbf{Ikeda map}: a canonical low-dimensional discrete-time chaotic map~\citep{ikeda1979multiple},
\begin{align*}
x_{n+1} &= 1+\upsilon(x_n\cos t_n-y_n\sin z_n),\\
y_{n+1} &= \upsilon(x_n\sin t_n+y_n\cos z_n),\\
z_n &= 0.4-\frac{6}{1+x_n^2+y_n^2},
\end{align*}
with $\upsilon=0.9$, included to isolate discrete-time, one-shot chaotic iteration from the continuous-flow chaos of the Lorenz system above. Unlike a flow, the map possesses no intrinsic smoothness between successive iterates, allowing us to test whether continuous- and discrete-time deterministic chaos leave comparable or distinct signatures under the same evaluation procedure.

\item \textbf{NARMA10}: a standard nonlinear autoregressive moving-average benchmark of order 10,
\begin{equation}
y_{n+1} = \alpha\, y_n + \beta\, y_n \sum_{i=0}^{9} y_{n-i} + \gamma\, g_{n-9} g_n + \epsilon,
\end{equation}
with $g_n$ an exogenous i.i.d.\ input drawn from $\mathcal{U}(0, 0.5)$ and $(\alpha=0.3,\beta=0.05,\gamma=1.5,\epsilon=0.1)$ standard literature constants, widely used in the reservoir computing literature itself~\cite{atiya2000new,yan2024emerging} as a synthetic task with strong, but nonlinear and higher-order, temporal dependence driven by an exogenous stochastic input at every step. Its inclusion connects this evaluation directly to RC's own benchmarking conventions and provides the genuinely stochastic-driven counterpart to the four deterministic systems above.
\end{enumerate}

Each systems derived time-series had $20000$ data points, after excluding $2000$ data points as transient.
For each system, we construct ordinal-pattern sequences using a fixed embedding dimension $d$ and delay $\tau$, split each sequence into training and held-out test segments (70/30), and evaluate one-step-ahead prediction at many randomly sampled positions in the test segment, reporting mean and standard deviation across these samples.

\begin{figure}[htbp]
\centering
\includegraphics[width=\columnwidth]{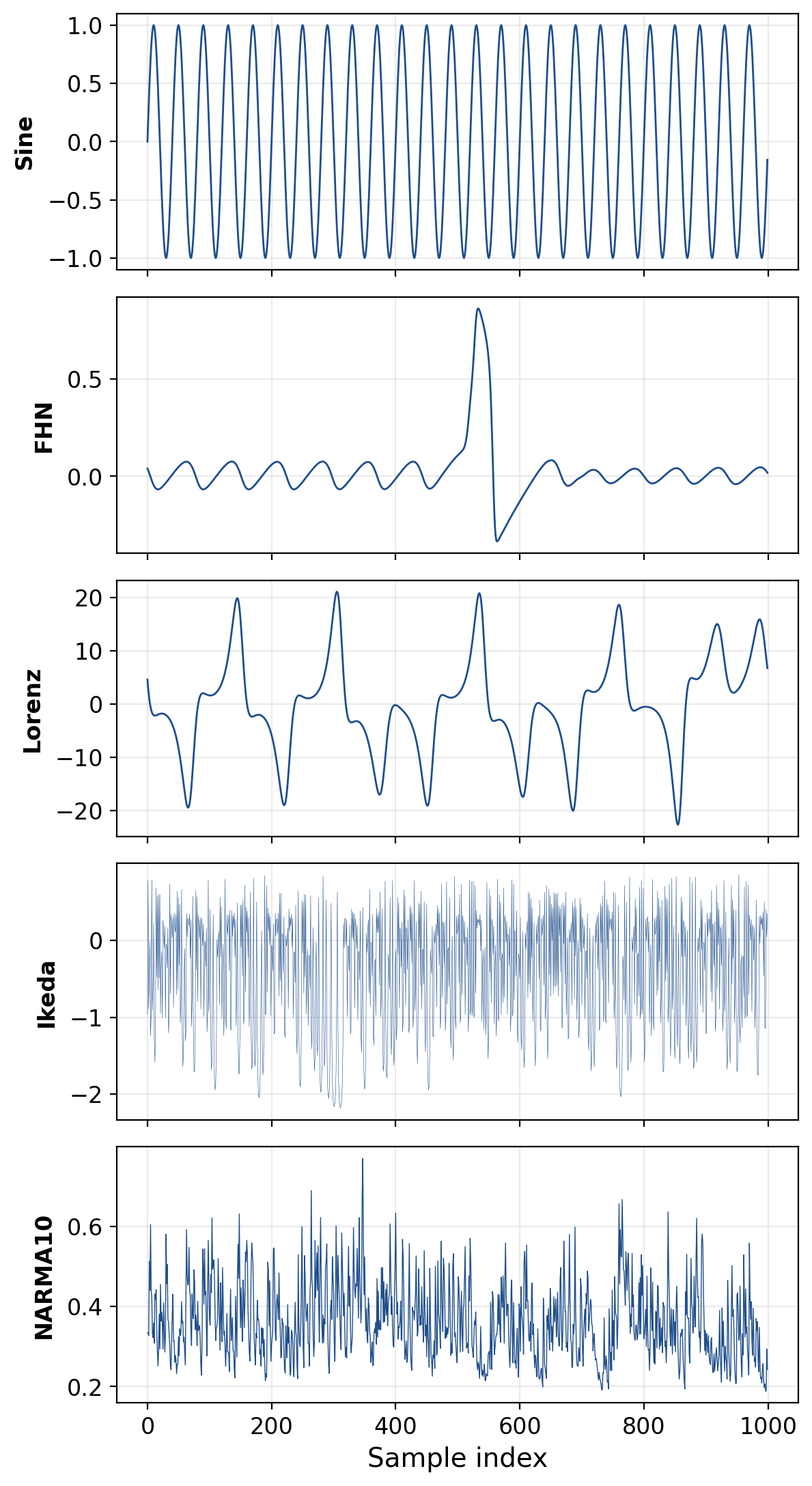}
\caption{Representative time-series segments for the five systems studied (Section~\ref{sec:systems}), each shown over a short window of 1000 consecutive samples. From top to bottom: the Sine wave's strict single-frequency periodicity; the coupled FHN network's ($x$-component) ensemble-averaged quasi-periodic oscillation, including one visible large-amplitude excursion consistent with an extreme event; the Lorenz system's characteristic chaotic switching between attractor lobes ($y$-component); the Ikeda map's aperiodic, sharply fluctuating trajectory ($y$-component) characteristic of discrete-time chaos in a two-dimensional attractor; and NARMA10's irregular, exogenously-driven fluctuation.}
\label{fig:example_series}
\end{figure}

\subsection{Results}
\label{sec:results}

Figure~\ref{fig:sweep_grid} summarizes one-step-ahead prediction performance across all five systems (Section~\ref{sec:systems}, from here on out also referred to as Sine, FHN, Lorenz, Ikeda, NARMA10), embedding dimensions $d \in \{3,4,5,6\}$, and delays $\tau \in \{1,\dots,100\}$, reporting exact accuracy, mean permutation distance, near-miss fraction, and the corrected chance baseline (Section~\ref{sec:pred_eval}) as the mean and standard deviation over ten realizations per system. We interpret these results not as a forecasting benchmark in their own right, but as an empirical probe of the fading-memory and substrate-structure properties established analytically in Sections~\ref{sec:central_correspondence}--\ref{sec:fading_memory}.

\begin{figure*}[t]
\centering
\includegraphics[width=\textwidth]{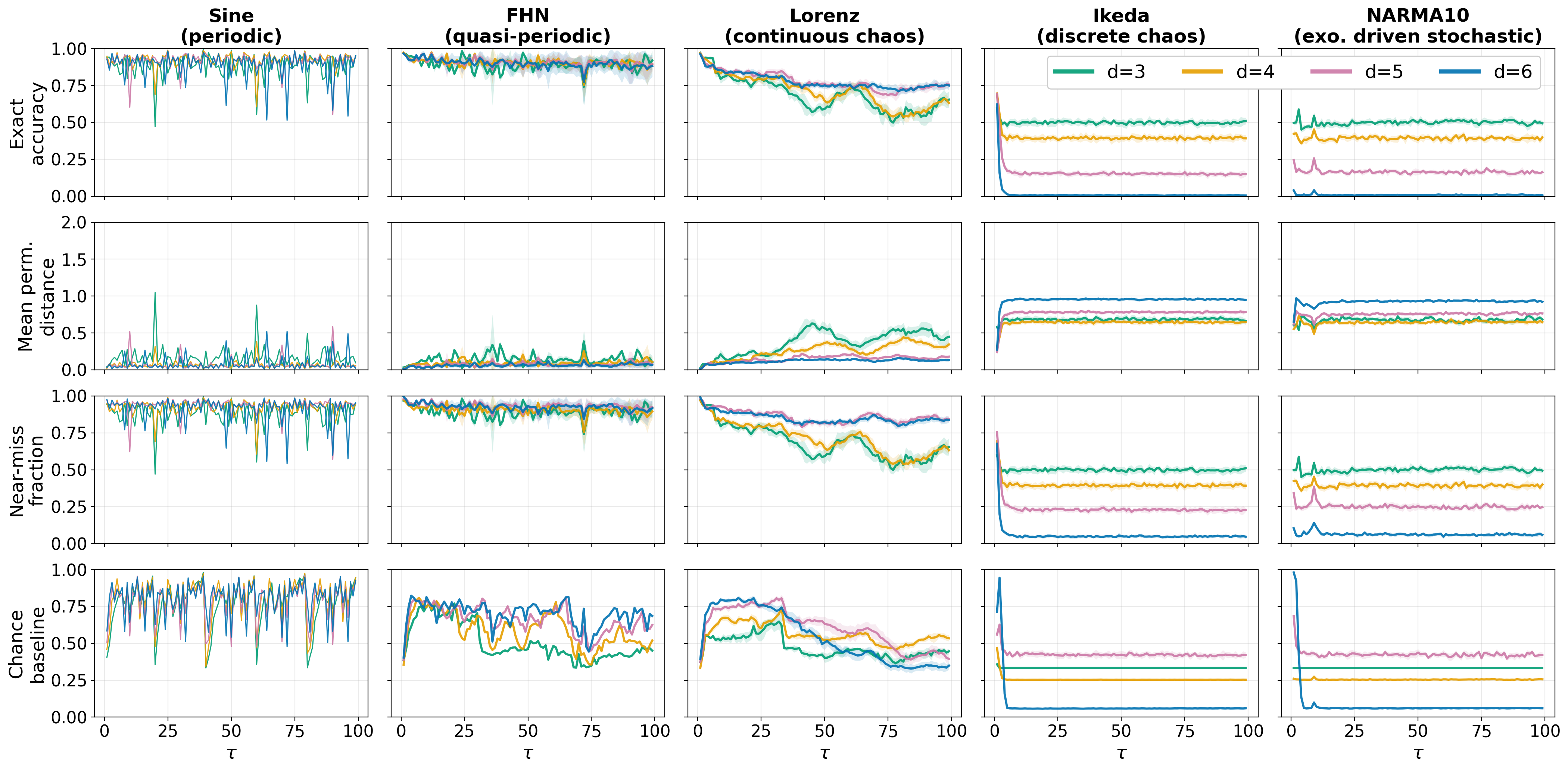}
\caption{One-step-ahead prediction statistics as a function of delay $\tau$, for embedding dimensions $d \in \{3,4,5,6\}$ (colors), across the five systems introduced in Section~\ref{sec:systems} (columns, left to right: Sine, FHN, Lorenz, Ikeda, NARMA10). Rows, top to bottom: exact accuracy, mean permutation distance, near-miss fraction, and the corrected chance baseline. Shaded bands denote $\pm 1$ standard deviation over ten realizations.}
\label{fig:sweep_grid}
\end{figure*}

\paragraph*{Regular and quasi-periodic dynamics yield near-ceiling accuracy.}
For the sine wave and the coupled FHN oscillators, both exact accuracy (typically exceeding $0.85$) and near-miss fraction remain near their ceiling values across the full range of delays examined, with little dependence on the embedding dimension $d$.
This is expected on theoretical grounds: the OPTN transition structure induced by strictly or quasi-periodic dynamics is highly constrained (Table~\ref{tab:correspondence}), so the history-augmented Markov chain of Definition~\ref{def:optn_chain} visits a small, predictable subset of $\Sigma_h$, and one-step-ahead prediction from the empirical transition table succeeds with high reliability.
A notable qualitative difference distinguishes the two systems: the sine wave shows sharp, non-monotonic decreases in accuracy at delays approaching rational fractions of its period, whereas the FHN oscillators show smoother variation with $\tau$.
This difference demonstrates that a high accuracy value alone does not identify a unique underlying mechanism; the corrected chance baseline (bottom row of Figure~\ref{fig:sweep_grid}) is necessary to disambiguate the two cases, since the sine wave's baseline varies in step with its accuracy -- an aliasing artifact of the substrate itself -- whereas the FHN oscillators' baseline varies far more smoothly, consistent with genuine, non-aliased dynamical structure.
This observation supports a general methodological point for OPTN-based random walks: accuracy values cannot be interpreted independently of the corrected chance baseline, because the OPTN substrate itself, not only the prediction rule, carries a delay- and dimension-dependent inductive bias.

\paragraph*{Lorenz dynamics show a memory-depth signature consistent with continuous-flow chaos.}
For the Lorenz system, exact accuracy decreases with increasing $\tau$, from approximately $0.9$ at small $\tau$ to approximately $0.5$--$0.6$ at $\tau = 100$, with larger embedding dimensions $d$ yielding consistently higher accuracy than smaller ones.
In the framework of Section~\ref{sec:fading_memory}, this behavior is qualitatively consistent with Proposition~\ref{prop:esp}: as $\tau$ increases, the correlation between consecutive embedding windows decreases, reducing the extent to which the history-augmented state constrains the distribution of the next pattern, and driving the induced chain $T_h$ toward a larger spectral gap and correspondingly shorter memory horizon, though we emphasize that the observed accuracy curve reflects the joint influence of the spectral gap together with transition sparsity, finite-sample estimation variance, and the chosen embedding dimension, rather than the spectral gap in isolation. The gradual decrease in exact accuracy with $\tau$, rather than an abrupt loss of accuracy, is itself informative: it reflects the continuous-flow property that samples separated by a small time delay remain numerically close, so the OPTN's edge structure degrades incrementally rather than discontinuously.
To our knowledge, this is the first empirical demonstration that the ESP-analogue spectral-gap diagnostic (Corollary~\ref{cor:mixing_time}) qualitatively tracks a system's known dynamical property -- sensitive dependence on initial conditions -- rather than reflecting an artifact of the symbolic encoding procedure.

To make this diagnostic concrete, we compute $\hat T_h$ directly from the Lorenz series at $d=4,\tau=5$ ($k=127$ recurrent history states), obtaining $\hat\lambda_\star=0.9052$ and spectral gap $\hat\gamma=0.0948$.
The closed-form bound of Corollary~\ref{cor:mixing_time}, evaluated using the constant $C=\sqrt{k}\hat\kappa(P)$, yields $t_{\mathrm{mix}}(0.05)\approx420$. Directly iterating $\hat T_h$ from 30 randomly sampled pairs of one-shot seed histories and measuring the resulting total-variation distance yields an empirical mixing time of $42$ steps, approximately one order of magnitude smaller (Figure~\ref{fig:empirical_mixing}).
This discrepancy reflects the numerical ill-conditioning of the eigenvector matrix $P$ for a sparse, empirically estimated stochastic matrix, which inflates $C$ substantially beyond the value realized in practice, rather than indicating a failure of the underlying spectral-gap mechanism: the exponential decay rate $\hat\lambda_\star^{\,n}$ correctly describes the asymptotic slope of the empirical decay curve, while the multiplicative prefactor $C$ does not meaningfully constrain the walk's practical memory depth.
We therefore recommend reporting the directly simulated empirical mixing time, together with the spectral gap $\hat\gamma$, as the more reliable diagnostics of memory depth in practice, and treat the closed-form bound of Corollary~\ref{cor:mixing_time} as qualitatively but not quantitatively informative.

\begin{figure}[htbp]
\centering
\includegraphics[width=\columnwidth]{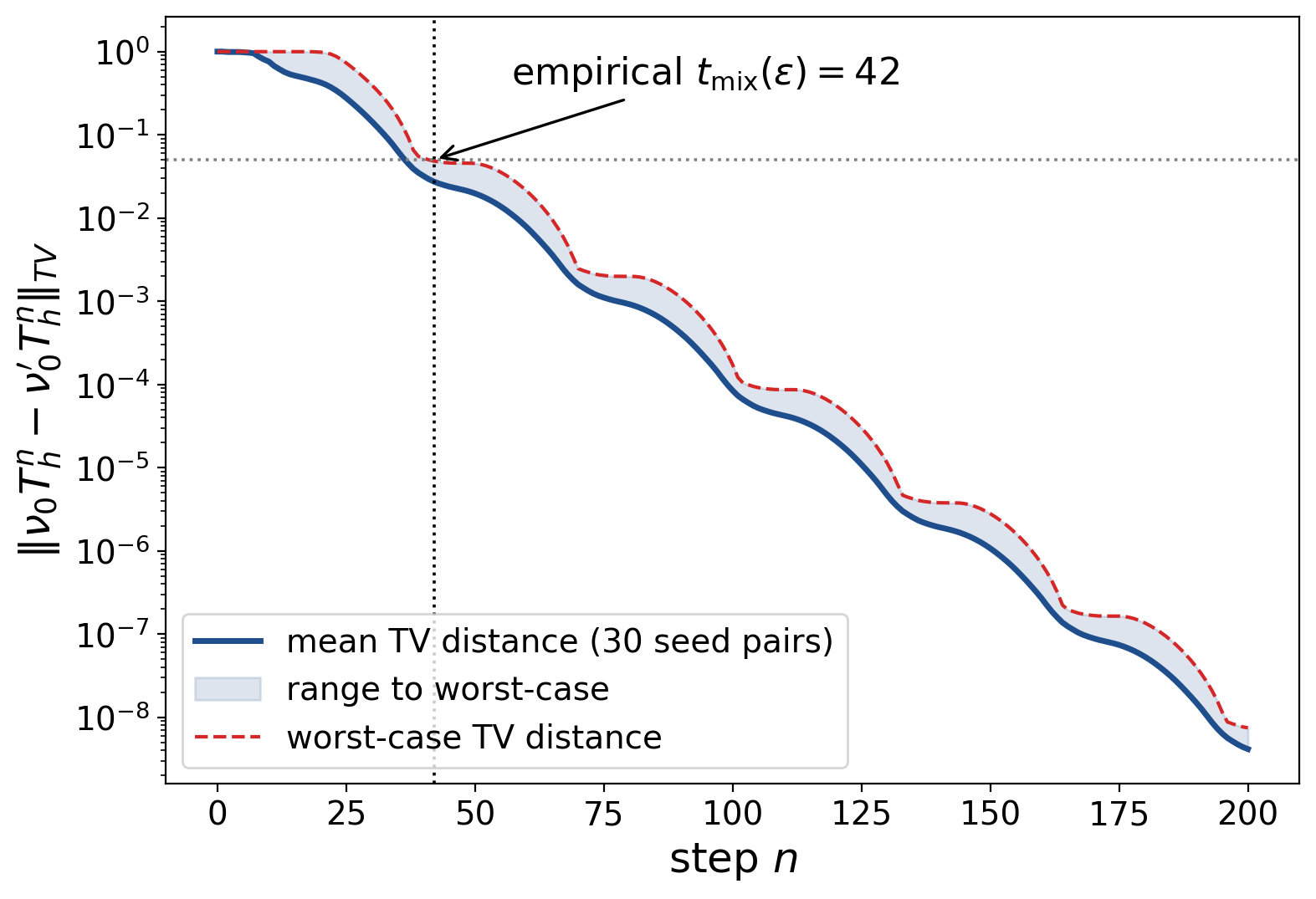}
\caption{Empirically simulated fading memory for the Lorenz system ($d=4,\tau=5$): total-variation distance $\|\nu_0 T_h^n-\nu_0'T_h^n\|_{TV}$ between pairs of one-shot seed histories, computed by direct iteration of the empirically estimated transition matrix $\hat T_h$ (30 random seed pairs; shaded band spans the mean to the worst-case value). The empirical mixing time $t_{\mathrm{mix}}(0.05)=42$ is approximately one order of magnitude smaller than the theoretical bound from Corollary~\ref{cor:mixing_time} ($\approx420$), reflecting the numerical looseness of the constant $C$ rather than a failure of the spectral-gap mechanism.}
\label{fig:empirical_mixing}
\end{figure}

\paragraph*{Discrete-time chaos and exogenous stochasticity produce a distinct, $\tau$-independent signature.}
For the Ikeda map and NARMA10, all four evaluation metrics are nearly indistinguishable between the two systems and differ substantially from those of the Lorenz system: exact accuracy, near-miss fraction, and the corrected chance baseline each saturate to a $\tau$-independent plateau within the first few delay steps, with the plateau value depending strongly on the embedding dimension $d$ but not on $\tau$ thereafter.
This result is notable because the two systems have unrelated generating mechanisms -- a deterministic one-dimensional map versus an exogenously driven stochastic autoregressive process -- and it bears directly on the correspondence developed in Table~\ref{tab:correspondence}.
We had hypothesized, on the grounds that both the Ikeda map and the Lorenz system are chaotic, that the Ikeda map's predictability profile would resemble that of the Lorenz system. The results instead indicate that OPTN substrate structure is determined not by whether the generating process is deterministic or stochastic, but by whether successive samples retain \emph{flow-like continuity}.
Continuous-time flows (Lorenz, FHN) preserve short-range ordinal correlations across small $\tau$, allowing the induced Markov chain $T_h$ to retain informative structure over a nontrivial range of delays; discrete-time, one-step iterated maps and exogenously driven stochastic processes both eliminate this continuity almost immediately, independent of whether the underlying generating rule is deterministic.
In reservoir-computing terms, this suggests that the effective memory capacity of the OPTN substrate is governed less by the entropy rate of the source signal than by its sample-to-sample regularity -- an empirically testable refinement, not previously reported to our knowledge, of what "predictability" means for a data-derived, symbolic reservoir.

\paragraph*{Chance-corrected accuracy is required to interpret exact accuracy meaningfully.}
Across all five systems, raw exact accuracy cannot be interpreted independently of the corrected chance baseline defined in Section~\ref{sec:pred_eval}.
For the Ikeda map and NARMA10 at low $d$, observed accuracy exceeds the corresponding baseline only modestly, whereas for the sine wave and FHN oscillators the OPTN's high connectivity raises the corrected baseline well above the naive uniform value of $1/d!$, so a moderate accuracy value can still represent a substantial improvement over chance.
We consider this baseline correction a necessary companion diagnostic for any subsequent work reporting accuracy for OPTN-based symbolic random walks, analogous to the role of null models in conventional reservoir-computing memory-capacity studies.
\tcom{We reiterate that these one-step-ahead prediction results are intended to probe the fading-memory and substrate-structure properties established analytically in Sections III–IV, rather than a competitive forecasting benchmark against optimized, task-tuned reservoir-computing architectures.}

\section{Discussion}
\label{sec:discussion}

We have argued that reservoir computing does not require a continuous state space: constrained random walks on ordinal pattern transition networks satisfy its core architectural and dynamical criteria, namely a fixed nonlinear substrate (Table~\ref{tab:correspondence}), a driven-to-autonomous operating mode consistent with RC-based digital twins \cite{kong2023reservoir}, and a fading-memory property with a precise discrete analogue in Markov chain ergodicity (Section~\ref{sec:fading_memory}, Proposition~\ref{prop:esp}).
Section~\ref{sec:empirical} further shows that even a simple, parameter-free maximum-likelihood rule suffices to perform a canonical RC task, one-step-ahead prediction, using only the fixed symbolic dynamics generated by the random walk.
This finding places OPTN-based frameworks within the established reservoir-computing framework, rather than treating them as an isolated symbolic-dynamics technique.
It also suggests that RC tools such as mixing-time estimates, memory-capacity measures, and spectral-gap analysis may be useful in ordinal-pattern-based analysis of complex dynamics.

Conversely, the correspondence suggests a promising direction for RC itself: reservoirs whose internal coupling structure is \emph{inferred from data} or even the \emph{dynamics of the underlying system} rather than randomly generated or physically fixed~\cite{koster2023data,geier2025dynamics}.
Existing calls within the RC community to explore ``structured\ldots non-random weights'' \cite{yan2024emerging} and to develop reservoirs with improved interpretability \cite{yan2024emerging} point toward reservoirs of different topologies~\cite{kawai2019small,houben2026reservoir} and exactly the kind of proposed data-derived, symbolic substrate.
Because the OPTN is built directly from empirically observed ordinal statistics, its structure is inherently interpretable -- each node corresponds to a named ordinal ordering of the dynamics -- and is therefore dynamics-informed in a way that random reservoir matrices are not.

Several open questions follow directly from this framing. Does the OPTN's ergodicity-based memory bound (Proposition~\ref{prop:esp}) relate in a simple and computable way to established RC memory-capacity measures?
Extending the fixed, parameter-free predictor of Section~\ref{sec:pred_rule} to a genuinely \emph{trained} readout -- for instance, a ridge or multinomial logistic regression over the history-augmented state $\Sigma_h$ -- is a natural and tractable next step, and would let us ask whether such an OPTN-derived readout inherits training guarantees comparable to Tikhonov-regularized linear readouts \cite{yan2024emerging}.

A related question concerns what the OPTN walk preserves beyond one-step-ahead accuracy: trained continuous reservoirs have been shown to remain topologically conjugate to the chaotic systems they are trained on, reproducing invariant geometric and statistical measures -- such as correlation dimension and multiscale entropy -- even as pointwise
trajectory accuracy degrades over longer horizons~\cite{chen2020mapping}. Whether an OPTN-based symbolic walk analogously preserves such invariants as one-step accuracy falls with increasing delay $\tau$ (Section~\ref{sec:results}) is untested here and would provide a more stringent assessment of the correspondence developed in Section~\ref{sec:central_correspondence} than one-step prediction accuracy alone.

Do other symbolic-dynamics-based approaches, such as recurrence networks and coarse-grained partition transition graphs admit similar reservoir-computing interpretations, suggesting a broader family of ``symbolic reservoirs" beyond the ordinal-pattern case considered here?
We view the present perspective as establishing the conceptual vocabulary needed to pursue these questions, with the construction of explicit trained readout mechanisms and their application to concrete time-series generation and surrogate or proxy data tasks left as directions for future work.
\vspace{-4mm}
\begin{acknowledgments}
The authors would like to thank Klaus Lehnertz, Manuel Adams, Simon Lehmann and
Bulcs\'u S\'andor for insightful and fruitful discussions on earlier
versions of the manuscript. The authors also thank the organizers of the \textit{Half a Century of the Kuramoto Model: Explaining Emergent Order in Time and Space} conference--Edmilson Roque dos Santos, Simona Olmi, Marc Timme, and Michael Rosenblum-- and the Wilhelm and Else Heraeus Foundation for providing a forum that enabled the authors based in Japan and Germany to connect and establish this collaboration.
Timo Br\"ohl and Max Potratzki gratefully acknowledge support and funding by the TRA Modelling (University of Bonn) as part of the Excellence Strategy of the federal and state governments.
\end{acknowledgments}

\section*{Author Declarations}
\subsection*{Author Contributions}
\textbf{Max Potratzki}: Conceptualization, Formal analysis, Investigation, Methodology, Software, Validation, Writing -- original draft, Writing -- review \& editing.
\textbf{Thomas G. de Jong}: Formal analysis, Validation (mathematics), Writing -- review \& editing.
\textbf{Michael Small}: Conceptualization,  Validation, Writing -- review \& editing.
\textbf{Timo Br\"ohl}: Conceptualization, Formal analysis, Funding acquistion, Investigation, Methodology, Project administration, Software, Supervision, Visualization, Writing -- original draft, Writing -- review \& editing.

\subsection*{Conflict of Interest}
The authors have no conflicts of interest to disclose.

\section*{Data Availability}
No new data were generated.

\bibliographystyle{aipnum4-1}
\bibliography{new}

\end{document}